\documentclass[graybox]{svmult}
\usepackage{mathptmx}       % selects Times Roman as basic font
\usepackage{helvet}         % selects Helvetica as sans-serif font
\usepackage{courier}        % selects Courier as typewriter font
\usepackage{makeidx}         % allows index generation
\usepackage{graphicx}        % standard LaTeX graphics tool
\usepackage{multicol}        % used for the two-column index
\usepackage[bottom]{footmisc}% places footnotes at page bottom
\usepackage{amsmath,amssymb,amsfonts}        % AMS math packages
\usepackage{algorithm}
\usepackage{algorithmicx}  
\usepackage{algpseudocode}
\usepackage{subfigure}
\makeindex             % used for the subject index
\begin{document}
\title*{Multi-Factor Polynomial Diffusion Models and Inter-Temporal Futures Dynamics}
\titlerunning{Multi-Factor Polynomial Diffusion Models}
\author{Peilun He, Nino Kordzakhia, Gareth W. Peters, Pavel V. Shevchenko}
\institute{Peilun He \at Macquarie University, Macquarie Park, NSW 2109, Australia, \email{peilun.he@mq.edu.au}
\and Nino Kordzakhia \at Macquarie University, Macquarie Park, NSW 2109, Australia, \email{nino.kordzakhia@mq.edu.au}
\and Gareth W. Peters \at University of California Santa Barbara, Santa Barbara, CA 93106, United States, \email{garethpeters@ucsb.edu}
\and Pavel V. Shevchenko \at Macquarie University, Macquarie Park, NSW 2109, Australia, \email{pavel.shevchenko@mq.edu.au}}
\maketitle

\abstract{In stochastic multi-factor commodity models, it is often the case that futures prices are explained by two latent state variables which represent the short and long term stochastic factors. In this work, we develop the family of stochastic models using polynomial diffusion to obtain the unobservable spot price to be used for modelling futures curve dynamics. The polynomial family of diffusion models allows one to incorporate a variety of non-linear, higher-order effects, into a multi-factor stochastic model, which is a generalisation of Schwartz and Smith \cite{schwartz2000short-term} two-factor model. Two filtering methods are used for the parameter and the latent factor estimation to address the non-linearity. We provide a comparative analysis of the performance of the estimation procedures. We discuss the parameter identification problem present in the polynomial diffusion case, regardless, the futures prices can still be estimated accurately. Moreover, we study the effects of different methods of calculating matrix exponential in the polynomial diffusion model. As the polynomial order increases, accurately and efficiently approximating the high-dimensional matrix exponential becomes essential in the polynomial diffusion model. }

\section{Introduction}
\label{sec:introduction}

Stochastic models employed in the analysis of commodity futures play a important role in various financial areas, including price forecasting, risk management, portfolio optimisation. In contrast to other commodities that are traded in both spot and futures markets, crude oil is primarily traded in the futures market. As a consequence, it becomes impractical to directly estimate the price of crude oil futures based on the spot price. Instead, people usually model the underlying spot price, denoted as $S_t$, as a function of certain factors. Under the assumption of an arbitrage-free market, the futures price at current time $t$, denoted as $F_{t,T}$, is equivalent to the expected spot price at maturity time $T$: 
\begin{equation}
    F_{t,T} = \mathbb{E}^*(S_T | \mathcal{F}_t), 
\end{equation}
where $\mathcal{F}_t$ be a natural $\sigma$-algebra generated up to time $t$ and $\mathbb{E}^*(\cdot)$ is the expectation taken with respect to the risk-neutral processes. Under this framework, the derivation of a closed-form expression of $F_{t,T}$ necessitates an accurate distribution of $S_t$. Consequently, this requirement imposes additional constraints on the factors involved in the modelling process. 

Over the past few decades, stochastic processes have been employed to model the factors. In 1990, the Ornstein-Uhlenbeck (OU) process was introduced for the modelling of oil futures in a two-factor setup to represent spot price and convenience yield \cite{gibson1990stochastic}. Building upon this work, Schwartz and Smith \cite{schwartz2000short-term} modelled the logarithm of the underlying spot price of crude oil futures as the sum of two hidden factors. These factors, assumed to follow the OU process, capture short-term fluctuation and long-term equilibrium price level, respectively. Subsequently, this latent factor model and its extensions became widely utilized in stochastic modelling. 

Researchers have further extended this model to enhance its applicability. In the electricity market, a multi-factor model including a deterministic seasonality with additional stochastic factors were modelled by Levy processes in \cite{eydeland1999fundamentals}. A time-changed Levy process were commonly used in option pricing to describe the jump behaviour and dynamics of volatility in \cite{carr2004time, fallahgoul2020risk, huang2004specification}. Sorensen \cite{sorensen2002modeling} extended the model by introducing three hidden factors, including an additional deterministic seasonal component, to capture the dynamics of agricultural commodity prices. In \cite{kiesel2009two}, the authors focused on direct modelling of the electricity futures prices instead of modelling of electricity spot price first. Ames et al. \cite{ames2020risk} incorporated time-varying drift and speed of mean reversion parameters in their modelling of crude oil futures. Favetto and Samson \cite{favetto2010parameter} applied this model in the field of biology, and used both maximum likelihood and expectation maximization methods for parameter estimations. Further, the performance of the multi-factor model in deriving spot prices was improved by incorporating the analyst's forecasts for futures prices, as demonstrated in \cite{cortazar2019commodity}. Peters et al. \cite{peters2013calibration} developed a partial Markov Chain Monte Carlo method to deal with the non-linear non-Gaussian multi-factor model. Comparing different models is also of significance. Schwartz \cite{schwartz1997the} compared models with up to three factors, including a hidden factor, convenience yield and interest rate, for copper, oil and gold. Cortazar and Naranjo \cite{cortazar2006an} compared the performance of one- to four-factor models in crude oil futures, and found that the three-factor and four-factor models excelled in explaining the term structure of futures prices, while the four-factor model outperformed others in fitting the volatility term structure.   

While this framework has gained popularity in the past two decades, it does possess certain limitations. Firstly, as previously mentioned, obtaining a closed-form expression of futures price necessitates an accurate distribution of spot price $S_t$. To address this issue, it is common practice to assume that all factors involved follow a Gaussian distribution, and the logarithm of spot price is a linear function of these factors. This assumption ensures a log-normal distribution of the spot price. Secondly, under this framework, it is common to model the logarithm of spot price. In most instances, this poses no significant problems. However, a noticeable event occurred on 20th April 2020 when the front-month May 2020 WTI crude oil futures settled at a unprecedented value of -\$37.63 per barrel on the New York Mercantile Exchange. This exceedingly rare phenomenon significantly changes the validity of the entire framework. 

This paper aims to address the aforementioned limitations by introducing a polynomial diffusion framework that allows for a more complicated structure of the spot price. The mathematical foundations were introduced in \cite{filipovic2016polynomial}. Under the polynomial diffusion framework, the spot price is represented as a polynomial of any order in terms of the factors. In particular, under certain conditions, it can be proven that the conditional expectation of spot price, which is equivalent to the futures price under the assumption of an arbitrage-free market, is also a polynomial in terms of factors. An application of this framework can be found in the modelling of electricity forwards, where the spot price is represented by a quadratic form of two factors \cite{kleisingeryu2019a}. Additionally, in this study, the Quadratic Kalman Filter was employed to estimate model parameters and unknown state variables. The state space was augmented to include both linear and quadratic terms of the factors. However, as the order of the polynomial increases, the dimension of the state space grows exponentially, making it challenging to derive the explicit form of the augmented state equations. To address this issue, we propose the use of the Extended Kalman Filter (EKF) and Unscented Kalman Filter (UKF). These methods provide effective tools for estimating parameters and the state variables within the context of the polynomial diffusion framework. The EKF and UKF help overcome the computational challenges associated with higher-order polynomial models, enabling the practical implementation of this framework. 

This paper is structured as follows. In Sect.~\ref{sec:ss_model} and Sect.~\ref{sec:pd_model}, we present the two aforementioned frameworks for the pricing of commodity futures. The first framework extends the Schwartz-Smith two-factor model \cite{schwartz2000short-term}, while the second framework models the spot price using polynomial forms. In the second framework, polynomial diffusion is employed to price the futures contracts. Sect.~\ref{sec:filter} introduces the Extended Kalman Filter (EKF) and Unscented Kalman Filter (UKF) as estimation methods for the hidden factors and unknown parameters in the polynomial diffusion models. These filters are specifically designed to handle the non-linear dynamics present in the models. Sect.~\ref{sec:results} presents a numerical analysis of the applications of the polynomial diffusion model. Firstly, we compare seven different methods for calculating the matrix exponential, which is required in the polynomial diffusion model. Our results indicate that the eigen-decomposition method provides an efficient and accurate approximation of the matrix exponential. Next, we evaluate the performance of the polynomial diffusion model through a simulation study. While the futures contracts can be estimated accurately, parameter estimation remains challenging, even when separating parameters in the state equation and the coordinate representations in the measurement equation. As a consequence, selecting the order of the polynomial diffusion model proves to be a challenging task, and further study on constraints is required. Finally, Sect.~\ref{sec:conclusion} concludes the paper. 

\section{Schwartz-Smith Two-Factor Model}
\label{sec:ss_model}
	
In this section, we describe a classical approach to modelling commodity futures, which is an extension of \cite{schwartz2000short-term}. 
	
This approach models the logarithm of spot price $S_t$ as the sum of two unobservable factors $\chi_t$ and $\xi_t$, 
\begin{equation}
\log{(S_t)} = \chi_t + \xi_t,
\label{eq:SS_st}
\end{equation}
where $\chi_t$ represents the short-term fluctuation and $\xi_t$ is the long-term equilibrium price level. We assume both $\chi_t$ and $\xi_t$ follow an OU process, 
\begin{equation}
d\chi_t = -\kappa \chi_t dt + \sigma_{\chi} d W_t^{\chi},
\label{eq:SS_chi}
\end{equation}
and 
\begin{equation}
d\xi_t = (\mu_{\xi} - \gamma \xi_t)dt + \sigma_{\xi} dW_t^{\xi},
\label{eq:SS_xi}
\end{equation}
while in \cite{schwartz2000short-term} only one factor follows the OU process. We assume the changes in the short-term factor $\chi_t$ are temporary and converging to 0 as $t \to \infty$. The processes $(W_t^{\chi})_{t \ge 0}$ and $(W_t^{\xi})_{t \ge 0}$ are correlated standard Brownian Motions with correlation coefficient $\rho$. Here, $\kappa, \gamma \in \mathbb{R}^+$ are the speed of mean-reversion parameters; $\mu_{\xi} \in \mathbb{R}$ is the mean level of the long-term factor; $\sigma_{\chi}, \sigma_{\xi} \in \mathbb{R}^+$ are the volatility parameters; and $\lambda_{\chi}, \lambda_{\xi} \in \mathbb{R}$ are risk premiums. 
	
By assuming a constant risk premium $\lambda_{\chi}$ and $\lambda_{\xi}$, the risk-neutral processes of $\chi_t$ and $\xi_t$ are given by
\begin{equation}
d\chi_t = (-\kappa \chi_t - \lambda_{\chi}) dt + \sigma_{\chi} d W_t^{\chi*}, 
\label{eq:SS_rn_chi}
\end{equation}
and 
\begin{equation}
d\xi_t = (\mu_{\xi} - \gamma \xi_t - \lambda_{\xi})dt + \sigma_{\xi} dW_t^{\xi*}, 
\label{eq:SS_rn_xi}
\end{equation}
where $W_t^{\chi*}$ and $W_t^{\xi*}$ are correlated standard Brownian Motions with correlation coefficient $\rho$. This approach stems from the risk-neutral futures pricing theory developed in \cite{black1976the}. 

Let $\mathcal{F}_t$ be a natural $\sigma$-algebra generated up to time $t$. In discrete time, given the initial values $\chi_{t_0}$ and $\xi_{t_0}$, $\chi_t$ and $\xi_t$ are jointly normally distributed with mean 
$$\mathbb{E}^*\left( \left. \left[ \begin{matrix}
    \chi_t \\
    \xi_t
\end{matrix} \right] \right| \mathcal{F}_{t_0} \right) = 
\left[ \begin{matrix}
    e^{-\kappa (t-t_0)}\chi_{t_0} - \frac{\lambda_{\chi}}{\kappa} \left(1 - e^{-\kappa (t-t_0)}\right) \\
    e^{-\gamma (t-t_0)}\xi_{t_0} + \frac{\mu_{\xi} - \lambda_{\xi}}{\gamma} \left(1 - e^{-\gamma (t-t_0)}\right) 
\end{matrix} \right]$$
and covariance matrix
$$Cov^*\left( \left. \left[ \begin{matrix}
    \chi_t \\
    \xi_t
\end{matrix} \right] \right| \mathcal{F}_{t_0} \right) =  \left[\begin{matrix}
\frac{1 - e^{-2\kappa (t-t_0)}}{2\kappa} \sigma_{\chi}^2 & \frac{1 - e^{-(\kappa + \gamma) (t-t_0)}}{\kappa + \gamma} \sigma_{\chi} \sigma_{\xi} \rho \\
\frac{1 - e^{-(\kappa + \gamma) (t-t_0)}}{\kappa + \gamma}\sigma_{\chi}\sigma_{\xi}\rho & \frac{1 - e^{-2\gamma (t-t_0)}}{2\gamma}\sigma_{\xi}^2
\end{matrix}\right], $$
where $\mathbb{E}^*(\cdot)$ and $Cov^*(\cdot)$ represent the expectation and covariance taken with respect to the risk-neutral processes. Therefore, the spot price, which is defined in \eqref{eq:SS_st}, is log-normally distributed with 
\begin{align}
\log[\mathbb{E}^*(S_t | \mathcal{F}_{t_0})] &= \mathbb{E}^*[\log(S_t) | \mathcal{F}_{t_0}] + \frac{1}{2}Var^*[\log(S_t) | \mathcal{F}_{t_0}] \nonumber\\
&= e^{-\kappa (t-t_0)}\chi_{t_0} + e^{-\gamma (t-t_0)}\xi_{t_0} + A(t-t_0), 
\label{eq:SS_logESt}
\end{align}	
where
\begin{align}
A(t) =& -\frac{\lambda_{\chi}}{\kappa}(1 - e^{-\kappa t}) + \frac{\mu_{\xi} - \lambda_{\xi}}{\gamma}(1 - e^{-\gamma t}) \nonumber \\
&+ \frac{1}{2}\left(\frac{1 - e^{-2\kappa t}}{2\kappa}\sigma_{\chi}^2 + \frac{1 - e^{-2\gamma t}}{2\gamma}\sigma_{\xi}^2 + 2\frac{1 - e^{-(\kappa + \gamma)t}}{\kappa + \gamma}\sigma_{\chi}\sigma_{\xi}\rho \right).
\end{align}

Next, we derive the equations for the futures prices. Let $F_{t, T}$ be the market price of a futures contract at time $t$ with maturity time $T$. For eliminating arbitrage, given all information until time $t$, the futures price must be equal to the expected spot price at the maturity time $T$. Therefore, under the risk-neutral measure, we have (assuming the interest rate is not stochastic)
$$\log{(F_{t, T})} = \log{[\mathbb{E}^*(S_T | \mathcal{F}_t)]} = e^{-\kappa (T-t)} \chi_t + e^{-\gamma (T-t)} \xi_t + A(T-t). $$
After discretization, we have the following AR(1) dynamics for bivariate state variable $x_t$
\begin{equation}
x_t = c + Ex_{t-1} + w_t, 
\label{eq:xt}
\end{equation}
where 
$$x_t = \left[ \begin{matrix} \chi_t \\ \xi_t \end{matrix} \right],\; c = \left[ \begin{matrix} 0 \\ \frac{\mu_{\xi}}{\gamma} \left(1 - e^{-\gamma \Delta t} \right) \end{matrix} \right],\; E = \left[ \begin{matrix} e^{-\kappa \Delta t} & 0 \\ 0 & e^{-\gamma \Delta t}\end{matrix} \right], $$
and $w_t$ is a column vector of correlated normally distributed noises with $\mathbb{E}(w_t) = \textbf{0}$ and
$$Cov(w_t) = \Sigma_w = \left[\begin{matrix}
\frac{1 - e^{-2\kappa \Delta t}}{2\kappa}\sigma_{\chi}^2 & \frac{1 - e^{-(\kappa + \gamma) \Delta t}}{\kappa + \gamma}\sigma_{\chi}\sigma_{\xi}\rho \\
\frac{1 - e^{-(\kappa + \gamma) \Delta t}}{\kappa + \gamma}\sigma_{\chi}\sigma_{\xi}\rho & \frac{1 - e^{-2\gamma \Delta t}}{2\gamma}\sigma_{\xi}^2
\end{matrix}\right],$$
$\Delta t$ is the time step between $(t-1)$ and $t$. Moreover, we have the measurement equation 
\begin{equation}
y_t = d_t + F_t^\top x_t + v_t, 
\label{eq:yt}
\end{equation}
where 
$$y_t = \left( \log{(F_{t,T_1})}, \dots, \log{(F_{t,T_m})} \right)^\top, d_t = \left( A(T_1 - t), \dots, A(T_m - t) \right)^\top, $$
$$F_t = \left[ \begin{matrix} e^{-\kappa (T_1 - t)}, \dots, e^{-\kappa (T_m - t)} \\ e^{-\gamma (T_1 - t)}, \dots, e^{-\gamma (T_m - t)} \end{matrix} \right], $$
and $m$ is the number of futures contracts. $v_t$ is an $m$-dimensional vector of normally distributed noises with $\mathbb{E}(v_t) = \textbf{0}$ and
$$Cov(v_t) = \Sigma_v = \left[ \begin{matrix} \sigma_1^2 & 0 & \dots & 0 \\ 0 & \sigma_2^2 & \dots & 0 \\ \vdots & \vdots & \ddots & \vdots \\ 0 & 0 & \dots & \sigma_m^2 \end{matrix} \right]. $$

The prediction error $e_t = y_t - \mathbb{E}(y_t|\mathcal{F}_{t-1})$ are supposed to be multivariate normally distributed. Therefore, the log-likelihood function of $y = (y_1, \dots, y_n)$ can be written as 
\begin{equation}
l(\theta;y) = -\frac{nm \log{(2\pi)}}{2} - \frac{1}{2} \sum_{t=1}^{n}{\left[ \log{[\det{(L_{t})}]} + e_t^\top L_{t}^{-1} e_t \right]}, 
\label{eq:ll}
\end{equation}
where the set of unknown parameters $\theta = (\kappa, \gamma, \mu_{\xi}, \sigma_{\chi}, \sigma_{\xi}, \rho, \lambda_{\chi}, \lambda_{\xi}, \sigma_1, \dots, \sigma_m)$; $n$ is the number of observations; $m$ is the number of contracts; $L_{t} = Cov(e_t)$. Given all observations $y$, the maximum likelihood estimate (MLE) of $\theta$ is obtained by maximising the log-likelihood function \eqref{eq:ll}. 

\section{Polynomial Diffusion Model}
\label{sec:pd_model}
	
In this section, we provide some important theorems of polynomial diffusions and apply these theorems to the two-factor model. The mathematical foundation and applications of polynomial diffusions in finance are provided in \cite{filipovic2016polynomial}. 
	
\begin{definition}
Consider the stochastic differential equation
\begin{equation}
dX_t = b(X_t)dt + \sigma(X_t)dW_t, 
\label{eq:sde}
\end{equation}
where $W_t$ is a $d$-dimensional standard Brownian motion and map $\sigma: \mathbb{R}^d \to \mathbb{R}^{d \times d}$ is continuous. Define $a := \sigma \sigma^\top$. For maps $a: \mathbb{R}^d \to \mathbb{S}^{d}$ and $b: \mathbb{R}^d \to \mathbb{R}^d$, suppose we have 
$a_{ij} \in Pol_2$ and $b_i \in Pol_1$. $\mathbb{S}^d$ is the set of all real symmetric $d \times d$ matrices and $Pol_n$ is the set of all polynomials of degree at most $n$. Then the solution of \eqref{eq:sde} is a polynomial diffusion. 
\end{definition}

Moreover, we define the generator $\mathcal{G}$ associated with the polynomial diffusion $X_t$ as
\begin{equation}
\mathcal{G}f(x) = \frac{1}{2} Tr\left( a(x) \nabla^2 f(x)\right) + b(x)^\top \nabla f(x)
\label{eq:generator}
\end{equation}
for $x \in \mathbb{R}^d$ and any $C^2$ function $f$. Let $N = C(d+n, n)$ be the dimension of $Pol_n$, and $H: \mathbb{R}^d \to \mathbb{R}^N$ be a function whose components form a basis of $Pol_n$. Then for any $p \in Pol_n$, there exists a unique vector $\vec{p} \in \mathbb{R}^N$ such that 
\begin{equation}
p(x) = H(x)^\top \vec{p}
\label{eq:vec_p}
\end{equation}
and $\vec{p}$ is the coordinate representation of $p(x)$. Moreover, there exists a unique matrix representation $G \in \mathbb{R}^{N \times N}$ of the generator $\mathcal{G}$, such that $G \vec{p}$ is the coordinate vector of $\mathcal{G} p$. So we have
\begin{equation}
\mathcal{G} p(x) = H(x)^\top G \vec{p}. 
\label{eq:G}
\end{equation}

\begin{theorem}
	Let $p(x) \in Pol_n$ be a polynomial with coordinate representation $\vec{p} \in \mathbb{R}^N$ satisfying \eqref{eq:vec_p}, $G \in \mathbb{R}^{N \times N}$ be a matrix representation of generator $\mathcal{G}$ satisfying \eqref{eq:G}, and $X_t \in \mathbb{R}^d$ satisfy \eqref{eq:sde}. Then for $0 \le t \le T$, we have
	$$\mathbb{E} \left[p(X_T)|\mathcal{F}_t \right] = H(X_t)^\top e^{(T-t)G} \vec{p}, $$
	where $\mathcal{F}_t$ represents all information available until time $t$. 
	\label{th:pd}
\end{theorem}

\begin{proof}
The proof is given in \cite{filipovic2016polynomial}. 
\qed
\end{proof}

Next, we apply this theorem to the two-factor model. Assume the spot price $S_t$ is modelled as
\begin{equation}
S_t = p_n(x_t), 
\label{eq:poly_st}
\end{equation}
where $x_t = (\chi_t, \xi_t)^\top$ is a vector of state variables and $p_n(\cdot)$ is a polynomial function with a degree at most $n$. $\chi_t$ and $\xi_t$ are the short-term and long-term factors defined in \eqref{eq:SS_chi} and \eqref{eq:SS_xi} for real processes and \eqref{eq:SS_rn_chi} and \eqref{eq:SS_rn_xi} for risk-neutral processes. Obviously, $x_t$ satisfies the stochastic differential equation \eqref{eq:sde}, with 
$$b(x_t) = \left[ \begin{matrix} -\kappa \chi_t - \lambda_{\chi} \\ \mu_{\xi} - \gamma \xi_t - \lambda_{\xi} \end{matrix} \right], \sigma(x_t) = \left[ \begin{matrix} \sigma_{\chi} & 0 \\ 0 & \sigma_{\xi} \end{matrix} \right], a(x_t) = \sigma(x_t) \sigma(x_t)^\top = \left[ \begin{matrix} \sigma_{\chi}^2 & 0 \\ 0 & \sigma_{\xi}^2 \end{matrix} \right]. $$
For any basis $H_n(x_t)$, the polynomial $p_n(x_t)$ can be uniquely represented as
$$p_n(x_t) = H_n(x_t)^\top \vec{p}. $$
The generator $\mathcal{G}$ is given by 
$$\mathcal{G}f(x) = \frac{1}{2} Tr \left( \left[ \begin{matrix} \sigma_{\chi}^2 & 0 \\ 0 & \sigma_{\xi}^2 \end{matrix} \right] \nabla^2 f(x) \right) + \left[ \begin{matrix} -\kappa \chi_t - \lambda_{\chi} \\ \mu_{\xi} - \gamma \xi_t - \lambda_{\xi} \end{matrix} \right]^\top \nabla f(x). $$
By applying $\mathcal{G}$ to each element of $H_n(x_t)$, we get the matrix representation $G$. Then, by Theorem \ref{th:pd}, the futures price $F_{t,T}$ is given by
\begin{equation}
F_{t,T} = \mathbb{E}^*(S_T | \mathcal{F}_t) = H(x_t)^\top e^{(T-t)G} \vec{p}. 
\label{eq:qua_ftt}
\end{equation}
Therefore, we have the non-linear state-space model 
\begin{equation}
x_t = c + E x_{t-1} + w_t, w_t \sim N(\textbf{0}, \Sigma_w), 
\label{eq:qua_xt}
\end{equation}
and 
\begin{equation}
y_t = H_n(x_t)^\top e^{(T-t)G} \vec{p} + v_t, v_t \sim N(\textbf{0}, \Sigma_v). 
\label{eq:qua_yt}
\end{equation}

In this paper, we consider a polynomial with degree 2, 
$$S_t = \alpha_1 + \alpha_2 \chi_t + \alpha_3 \xi_t + \alpha_4 \chi_t^2 + \alpha_5 \chi_t \xi_t + \alpha_6 \xi_t^2. $$
The basis of $Pol_2$ is 
$$H(x_t) = (1, \chi_t, \xi_t, \chi_t^2, \chi_t \xi_t, \xi_t^2)^\top, $$
which has a dimension $N=6$. The polynomial $S_t$ can be uniquely represented as
$$S_t = H(x_t)^\top \vec{p}, $$
where the coordinate representation $\vec{p}$ is given by 
$$\vec{p} = \left(\alpha_1, \alpha_2, \alpha_3, \alpha_4, \alpha_5, \alpha_6 \right)^\top. $$
Then, applying $\mathcal{G}$ to each element of $H(x_t)$, we get 
$$G = \left[ \begin{matrix} 
0 & -\lambda_{\chi} & \mu_{\xi} - \lambda_{\xi} & \sigma_{\chi}^2 & 0 & \sigma_{\xi}^2 \\ 
0 & -\kappa & 0 & -2 \lambda_{\chi} & \mu_{\xi} - \lambda_{\xi} & 0 \\ 
0 & 0 & -\gamma & 0 & -\lambda_{\chi} & 2\mu_{\xi} - 2\lambda_{\xi} \\ 
0 & 0 & 0 & -2\kappa & 0 & 0 \\ 
0 & 0 & 0 & 0 & -\kappa - \gamma & 0 \\ 
0 & 0 & 0 & 0 & 0 & -2\gamma
\end{matrix} \right]. $$

\section{Non-Linear Filtering Algorithm}
\label{sec:filter}
Kalman Filter (KF) is the most popular filtering method for estimating the state vector $x_t = (\chi_t, \xi_t)^\top$ based on a filtration $\mathcal{F}_t$. Moreover, the unknown parameters can be estimated jointly by maximising the log-likelihood function which is calculated using the KF. However, the KF can only deal with linear Gaussian state-space models. For dealing with the non-linearity of the polynomial diffusion model, we introduce two extensions of KF, the Extended Kalman Filter (EKF) and Unscented Kalman Filter (UKF). 

In this paper, we use the notation 
\begin{align}
a_{t|t-1} &:= \mathbb{E}(x_t | \mathcal{F}_{t-1}),& P_{t|t-1} &:= Cov(x_t | \mathcal{F}_{t-1}), \nonumber \\
a_t &:= \mathbb{E}(x_t | \mathcal{F}_t),& P_t &:= Cov(x_t | \mathcal{F}_t) \nonumber 
\end{align} 
to represent the expected values and covariance matrices. 

The first non-linear filtering method is EKF. In order to capture the non-linear dynamics in the polynomial diffusion model, EKF linearises the state and measurement equations. Consider the following non-linear dynamic system:
\begin{equation}
    x_t = f(x_{t-1}) + w_t, w_t \sim N(\textbf{0}, \Sigma_w), 
\end{equation}
\begin{equation}
    y_t = h(x_t) + v_t, v_t \sim N(\textbf{0}, \Sigma_v). 
\end{equation}
where $f(\cdot)$ and $h(\cdot)$ are non-linear functions; $x_t$ is the unobservable state vector; $y_t$ is the observation. The main idea of EKF is to linearise the functions $f(\cdot)$ and $h(\cdot)$ by the first-order Taylor series. Let $J_f$ and $J_h$ be the Jacobian of $f(\cdot)$ and $h(\cdot)$ respectively. 

Given the prior mean and a new observation at current time $t$, we calculate the posterior mean $a_{t|t-1}$ and prediction error $e_t$ as 
\begin{equation}
    a_{t|t-1} = f(a_{t-1}),
\end{equation}
\begin{equation}
    e_t = y_t - h(a_{t|t-1}). 
\end{equation}
However, because of the non-linearity in the state and measurement equations, we cannot get a closed-form of the posterior covariance $P_{t|t-1}$ and the covariance of prediction error $L_t$. In that case, we approximate $f(\cdot)$ and $h(\cdot)$ by first-order Taylor expansion at point $a_{t-1}$ and $a_{t|t-1}$ respectively: 
\begin{equation}
    f(x_{t-1}) = f(a_{t-1}) + J_f(a_{t-1}) (x_{t-1} - a_{t-1}), 
\end{equation}
\begin{equation}
    h(x_t) = h(a_{t|t-1}) + J_h(a_{t|t-1}) (x_t - a_{t|t-1}).
\end{equation}
Therefore,  $P_{t|t-1}$ and $L_t$ are calculated as: 
\begin{equation}
    P_{t|t-1} = J_f(a_{t-1}) P_{t-1} J_f^\top(a_{t-1}) + \Sigma_w, 
\end{equation}
\begin{equation}
    L_t = J_h(a_{t|t-1}) P_{t|t-1} J_h^\top(a_{t|t-1}) + \Sigma_v. 
\end{equation}
Finally, we update the prior mean and covariance at current time $t$ as: 
\begin{equation}
    a_t = a_{t|t-1} + K_t e_t, 
\end{equation}
\begin{equation}
    P_t = \left(I - K_t J_h(a_{t|t-1})\right) P_{t|t-1},
\end{equation}
where $K_t = P_{t|t-1} J_h^\top(a_{t|t-1}) L_{t}^{-1}$ is the Kalman gain matrix. 

%\begin{algorithm}[H]  
%\caption{Algorithm of Extended Kalman Filter} 
%\label{alg:EKF}
%\hspace*{0.02in} {\bf Input:} $a_0$, $P_0$, $y_t$ where $t \in [1,n]$ \\
%\hspace*{0.02in} {\bf Output:} $a_{t|t-1}$, $l_t$
%\begin{algorithmic}[1]  
%	\State $t = 1$
%	\For{$t \le n$} 
%	\State \textbf{Forecasting}: 
%	$$a_{t|t-1} = f(a_{t-1})$$
%	$$P_{t|t-1} = J_f(a_{t-1}) P_{t-1} J_f^\top(a_{t-1}) + \Sigma_w$$
%	\State \textbf{Innovation}:
%	$$e_t = y_t - h(a_{t|t-1})$$
%	$$L_t = J_h(a_{t|t-1}) P_{t|t-1} J_h^\top(a_{t|t-1}) + \Sigma_v$$
%	\State \textbf{Updating}: 
%	$$K_t = P_{t|t-1} J_h^\top(a_{t|t-1}) L_{t}^{-1}$$
%	$$a_t = a_{t|t-1} + K_t e_t$$
%	$$P_t = \left(I - K_t J_h(a_{t|t-1})\right) P_{t|t-1}$$
%	\State \textbf{Log-likelihood function}: 
%	$$l_t = - \log{\left(\det{(L_{t})}\right)} - e_t^\top L_{t}^{-1} e_t$$
%	\State $t= t + 1$
%	\EndFor
%	\State \Return $a_{t|t-1}$, $l_t$ 
%\end{algorithmic}  
%\end{algorithm}  

In the EKF, the state distribution is propagated by linearising the non-linear system using the first-order approximation. However, this linearisation process can introduce significant errors in the true state distribution, especially when the system exhibits a strong non-linearity. Additionally, obtaining an analytical Jacobian for complicated state and measurement equations may be impractical. To address these issues, we introduce a derivative-free filtering method called the Unscented Kalman Filter (UKF). Instead of linearising the system, the UKF employs a set of carefully selected points, known as sigma points, to represent the true distributions of the state variables. These sigma points are then propagated through the state equation. The true prior and posterior means, and covariance would be captured by the sigma points. 

At previous time $t-1$, the sigma points are defined as 
\begin{equation}
    \mathcal{X}_{t-1} = \left[ a_{t-1}, a_{t-1} \pm \sqrt{(n_x + \lambda)P_{t-1}} \right], 
\end{equation}
where $n_x$ is the number of state variables. In this paper, we set the scaling parameter $\lambda = 0$. A detailed description is available in \cite{wan2000unscented}. 

Next, these sigma points are propagated through the non-linear state and measurement equations: 
\begin{equation}
    \mathcal{X}_{t|t-1} = f(\mathcal{X}_{t-1}), 
\end{equation}
\begin{equation}
    \mathcal{Y}_{t|t-1} = h(\mathcal{X}_{t|t-1}). 
\end{equation}
The posterior mean $a_{t|t-1}$, one-step forecast $\hat{y}_{t|t-1}$ and covariances are weighted sample mean / covariance of sigma points. The remaining steps in the UKF are similar to those in the EKF.

\section{Results}
\label{sec:results}

In this section, we present the outcomes of our numerical experiments. In Sect.~\ref{sec:mat_exp}, we compare seven distinct numerical methods for the evaluation of a matrix exponential. Evaluating the matrix exponential of the matrix $G$ is necessary for obtaining a closed-form of futures prices according to Theorem \ref{th:pd}. With the increase in the degree of the polynomial, the dimension of $G$ grows exponentially. In Sect.~\ref{sec:sim_study}, we assess the performance of the polynomial diffusion model through a simulation study. 

\subsection{Matrix Exponential}
\label{sec:mat_exp}

Theorem~\ref{th:pd} gives a direct way of deriving the futures price with the spot price given in a polynomial form. However, the computation of matrix exponential $e^{(T-t)G}$ is required. As the degree of polynomial increases, the dimension of the $G$ matrix increases quickly. As a consequence, one must carefully choose the numerical methods to compute the matrix exponential. A discussion of different methods to compute matrix exponential is given in \cite{moler2003nineteen}. We compare different methods to compute $e^{A}$ in the following three aspects: 

\runinhead{Stability:}
If small changes in matrix $A$ cause large changes in $e^A$. It is evaluated by $$\phi = \frac{ ||e^{A+E} - e^A|| }{||e^A||}, $$ where $||\cdot||$ represents 2-norm and $E$ is a matrix whose norm is small. 

\runinhead{Accuracy:}
The difference between an approximation and the true value. As we know the vector of eigenvalues $\Lambda$ and the matrix of corresponding eigenvectors $V$ (as shown in the steps of generating matrix), the true value is calculated as $e^A = V e^\Lambda V^{-1}$. Let $B = \{b_{ij} \}$ be the matrix exponential calculated by one method and $C = \{ c_{ij} \}$ be the true value of the matrix exponential, then the accuracy is evaluated as $$\psi = \sum_{i=1}^{n} \sum_{j=1}^{n} (b_{ij} - c_{ij})^2. $$

\runinhead{Efficiency:}
The computing time required for this methods. 

~

Seven different methods are tested on 100 realisations of $10 \times 10$ random matrices. The matrix is generated by the following steps (where the dimension $n=10$): 
\begin{enumerate}
    \item{Generate a $n$-dimensional vector of eigenvalues $\Lambda = (\lambda_1, ..., \lambda_n)$, where $\lambda_i \sim N(0, 10^2)$. }
    \item{Generate a diagonal matrix $V$, with diagonal elements $\Lambda$. }
    \item{Generate a $n \times n$ matrix $U=\{ u_{ij} \}$, where $u_{ij} \sim N(0, 1)$. Then, divide each column of $U$ by its $L_2$ norm. Therefore, each column of $U$ forms an eigenvector. }
    \item{Calculate $A = U * V * U^{-1}$. }
\end{enumerate}

\begin{table}[ht]
    \caption{Stability, accuracy and efficiency of the listed methods. }
    \centering
    \label{tab:mat_exp}
    \begin{tabular}{cccc}
        \hline\noalign{\smallskip}
        Methods & Stability & Accuracy & Efficency \\
        \noalign{\smallskip}\svhline\noalign{\smallskip}
        Taylor series & 1.7184 & 10.0214 & 0.1038 \\
        Pade approximation & 1.7182 & 7.1447e+22 & 0.2415 \\
        Scaling and squaring & 1.7183 & 4.3527 & 0.0657 \\
        Lagrange & 1.7183 & 0.5412 & 0.1030 \\
        Newton & 1.7183 & 0.0938 & 0.2098 \\
        Vandermonde & 1.7183 & 4.0583e+10 & 0.1153 \\
        Eigen-decomposition & 1.7183 & 0.0094 & 0.0340 \\
        \noalign{\smallskip}\hline\noalign{\smallskip}
    \end{tabular}
\end{table}

%\begin{table}[ht]
%    \caption{Stability, accuracy and efficiency of the listed methods. }
%    \centering
%    \label{tab:mat_exp}
%    \begin{tabular}{cccc}
%        \hline\noalign{\smallskip}
%        Methods & Stability & Accuracy & Efficency \\
%        \noalign{\smallskip}\svhline\noalign{\smallskip}
%        Taylor series & 1.7183 & 31.8785 & 0.0925 \\
%        Pade approximation & 1.7183 & 1.4657e+21 & 0.2832 \\
%        Scaling and squaring & 1.7183 & 8.6253e+09 & 0.0624 \\
%        Lagrange & 1.7183 & 19.7442 & 0.0424 \\
%        Newton & 1.7183 & 18.2120 & 0.0923 \\
%        Vandermonde & 1.7183 & 3.7672e+06 & 0.0612 \\
%        Eigen-decomposition & 1.7183 & 20.2650 & 0.0260 \\
%        \noalign{\smallskip}\hline\noalign{\smallskip}
%    \end{tabular}
%\end{table}

The mean stability, mean accuracy and efficiency (in second) of seven methods are given in Table~\ref{tab:mat_exp}. All methods have similar stability. Eigen-decomposition method has the best accuracy which is only 0.0094. The Eigen-decomposition and Scaling and Squaring methods have the fastest computing time in ascending order. In conclusion, eigen-decomposition approximates the matrix exponential accurately and efficiently, and it is used in the following section to compute the matrix exponential. 

\subsection{Simulation Study}
\label{sec:sim_study}

In this section, we assess the performance of the polynomial diffusion model in terms of futures estimation and parameter estimation. To evaluate the model, we conducted a simulation study on two datasets. The first data has 1000 observations and 13 contracts, with maturity time from $T_1$ = 1 month for the first available contract to $T_{13}$ = 13 months for the last available contract. The second data also has 1000 observations but contains 20 contracts with maturity times up to 20 months. The time series of two data are given in Fig.~\ref{fig:price1} and Fig.~\ref{fig:price2}, and the term structures are given in Fig.~\ref{fig:price1_3D} and Fig.~\ref{fig:price2_3D}. 

\begin{figure}[ht]
	\centering
        \includegraphics[width=0.8\textwidth]{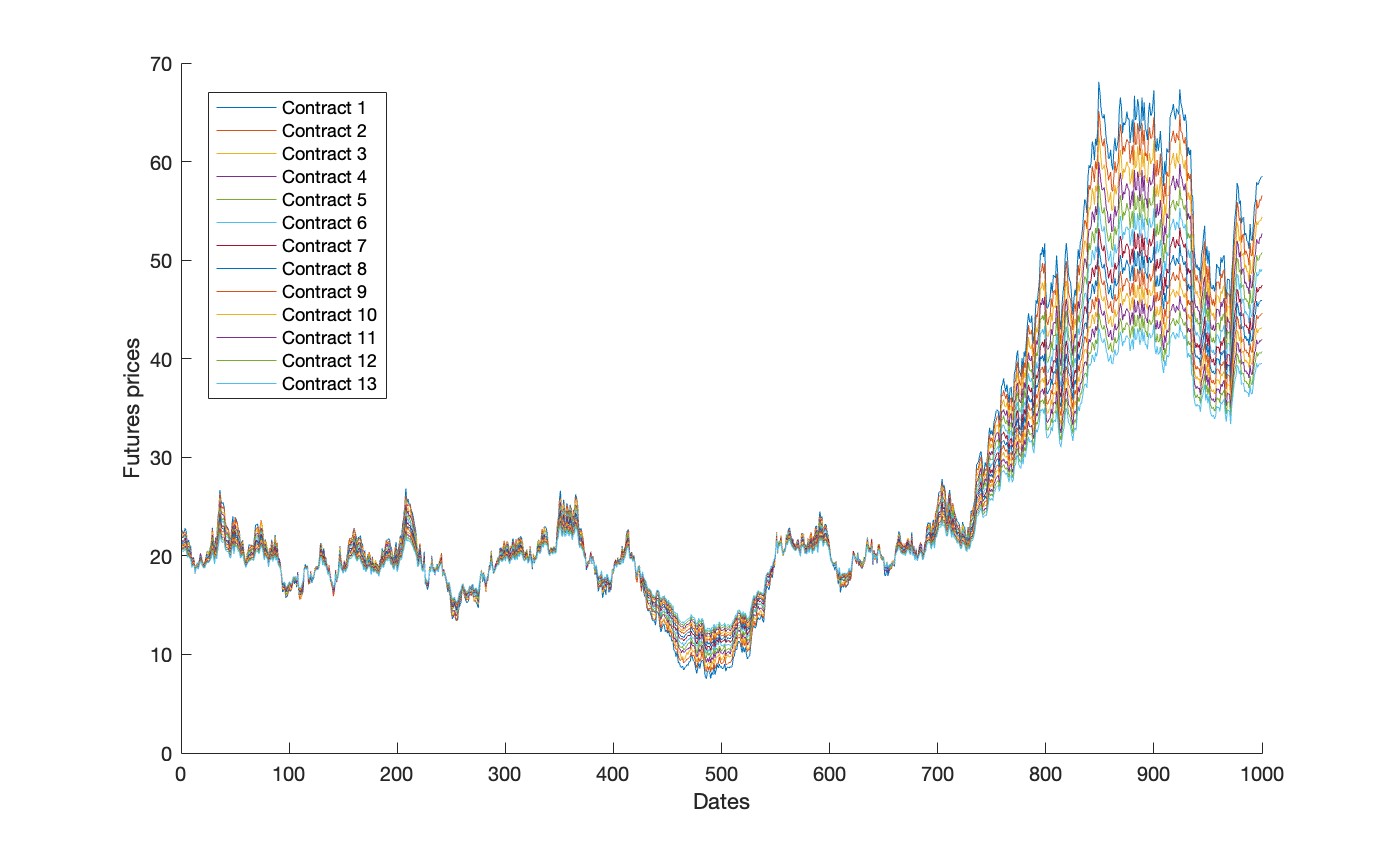}
	\caption{The time series plots of 13 simulated futures contract prices with the maturity times $T_1 = 1$ month, $T_2 = 2$ months, $\dots$, $T_{13}=13$ months respectively. }
        \label{fig:price1}
\end{figure}

\begin{figure}[ht]
	\centering
        \includegraphics[width=0.8\textwidth]{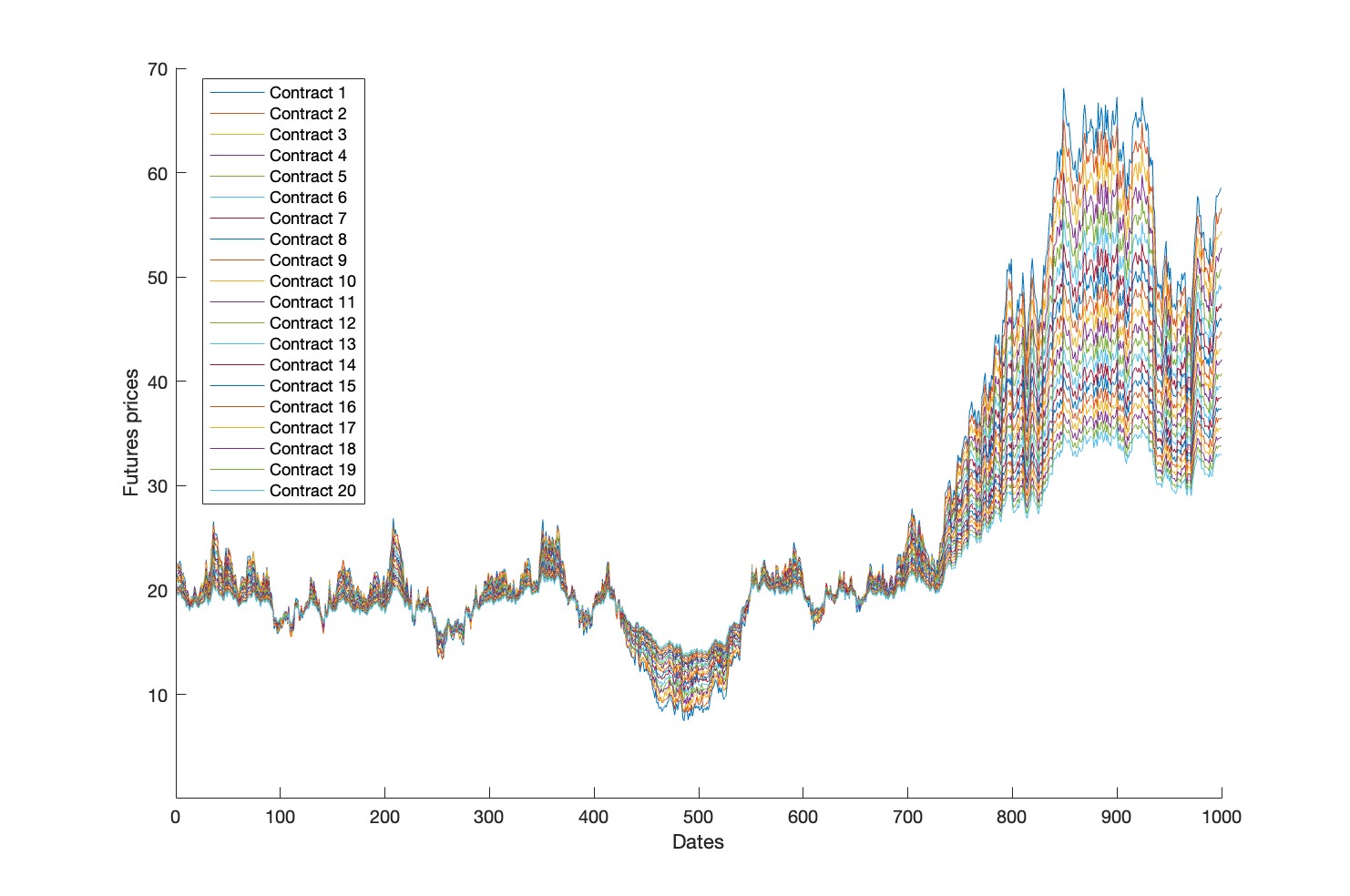}
	\caption{The time series plots of 20 simulated futures contract prices with the maturity times $T_1 = 1$ month, $T_2 = 2$ months, $\dots$, $T_{20}=20$ months respectively. }
        \label{fig:price2}
\end{figure}

\begin{figure}[ht]
	\centering
        \includegraphics[width=0.8\textwidth]{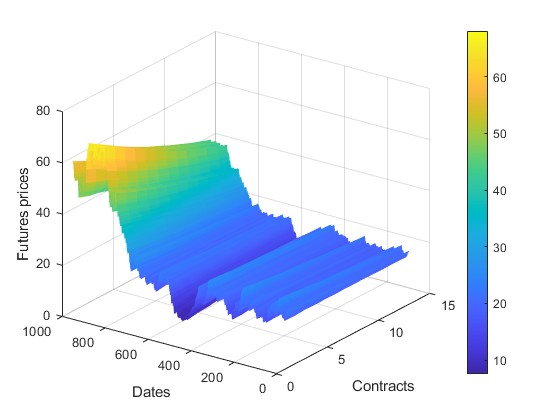}
	\caption{Term structure plot of 13 futures contracts prices versus their maturities. }
        \label{fig:price1_3D}
\end{figure}

\begin{figure}[ht]
	\centering
        \includegraphics[width=0.8\textwidth]{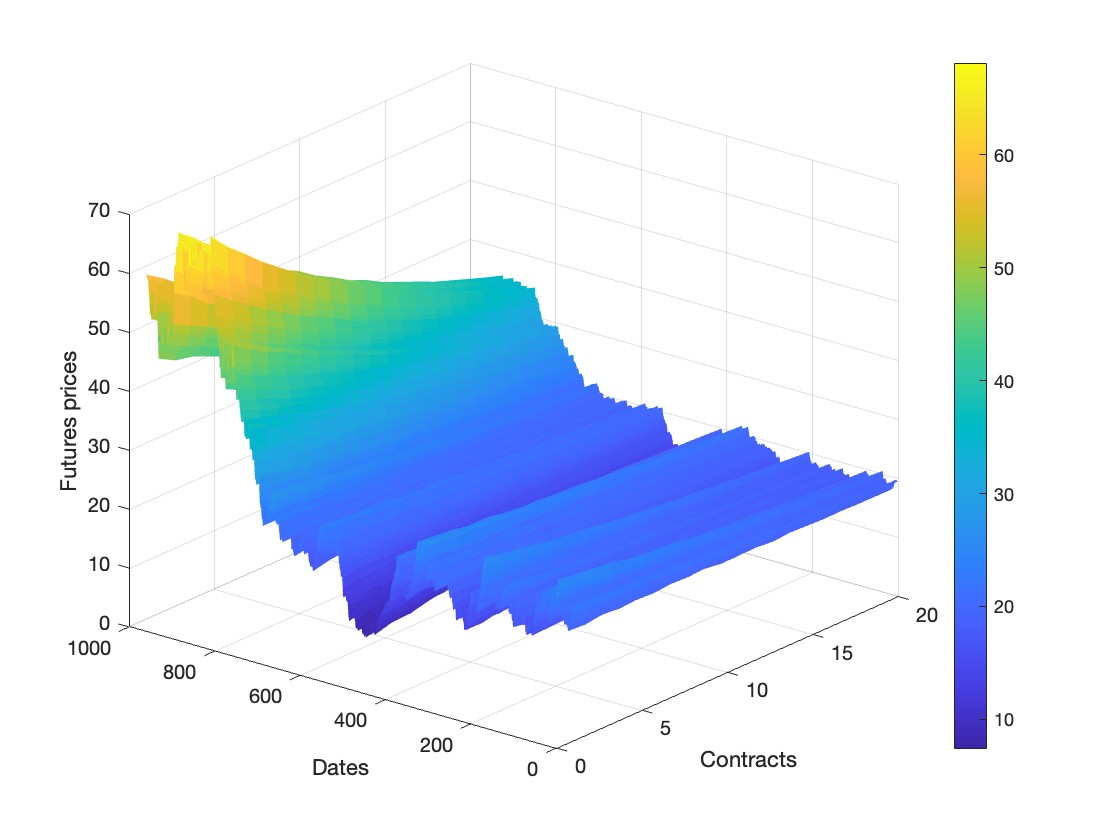}
	\caption{Term structure plot of 20 futures contracts prices versus their maturities. }
        \label{fig:price2_3D}
\end{figure}

Our target is to show that there exist some identification problems between model parameters and coordinate representation of spot price, but this would not affect the contract estimations. To achieve this, we first separate all parameters into two sets: model parameters $\theta = (\kappa, \gamma, \mu_{\xi}, \sigma_{\chi}, \sigma_{\xi}, \rho, \lambda_{\chi}, \lambda_{\xi}, \sigma_1, \sigma_2, \dots, \sigma_{20})$, which includes all state parameters and measurement errors, and coordinate representation $\vec{p} = (\alpha_1, \alpha_2, \alpha_3, \alpha_4, \alpha_5, \alpha_6)$. Then, the polynomial diffusion model performances were evaluated through root mean square error (RMSE) in the following 4 cases: 
\runinhead{Case 1.} All model parameters and coordinate representations are known. Only the futures contracts and state variables are estimated. 
\runinhead{Case 2.} Model parameters are fixed to the true values. 
\runinhead{Case 3.} Coordinate representations are fixed to the true values. 
\runinhead{Case 4.} All model parameters and coordinate representations are estimated. 

\begin{table}
\caption{Root mean square errors (RMSEs) for futures contracts in the case of EKF: comparison of 13-contract data.}
\label{tab:RMSE_EKF1}  
\centering
\begin{tabular}{p{3cm}p{2cm}p{2cm}p{2cm}p{1cm}}
\hline\noalign{\smallskip}
Contracts & Case 1 & Case 2 & Case 3 & Case 4 \\
\noalign{\smallskip}\svhline\noalign{\smallskip}
Contract 1 & 0.1200 & 0.1199 & 0.1182 & 0.1173 \\
Contract 2 & 0.1130 & 0.1129 & 0.1130 & 0.1150 \\
Contract 3 & 0.1079 & 0.1078 & 0.1084 & 0.1075 \\
Contract 4 & 0.0946 & 0.0946 & 0.0938 & 0.0910 \\
Contract 5 & 0.0893 & 0.0892 & 0.0895 & 0.0872 \\
Contract 6 & 0.0780 & 0.0779 & 0.0777 & 0.0769 \\
Contract 7 & 0.0704 & 0.0704 & 0.0706 & 0.0706 \\
Contract 8 & 0.0620 & 0.0620 & 0.0622 & 0.0621 \\
Contract 9 & 0.0506 & 0.0506 & 0.0506 & 0.0506 \\
Contract 10 & 0.0404 & 0.0404 & 0.0404 & 0.0412 \\
Contract 11 & 0.0305 & 0.0305 & 0.0304 & 0.0317 \\
Contract 12 & 0.0205 & 0.0205 & 0.0206 & 0.0239 \\
Contract 13 & 0.0100 & 0.0100 & 0.0098 & 0.0066 \\
Mean & 0.0683 & 0.0682 & 0.0681 & 0.0678 \\
\noalign{\smallskip}\hline\noalign{\smallskip}
\end{tabular}
\end{table}

\begin{table}
\caption{Root mean square errors (RMSEs) for futures contracts in the case of EKF: comparison of 20-contract data.}
\label{tab:RMSE_EKF2}  
\centering
\begin{tabular}{p{3cm}p{2cm}p{2cm}p{2cm}p{1cm}}
\hline\noalign{\smallskip}
Contracts & Case 1 & Case 2 & Case 3 & Case 4 \\
\noalign{\smallskip}\svhline\noalign{\smallskip}
Contract 1 & 0.1873 & 0.1873 & 0.1858 & 0.1939 \\
Contract 2 & 0.1802 & 0.1801 & 0.1803 & 0.1864 \\
Contract 3 & 0.1783 & 0.1783 & 0.1792 & 0.1827 \\
Contract 4 & 0.1603 & 0.1604 & 0.1593 & 0.1585 \\
Contract 5 & 0.1577 & 0.1578 & 0.1580 & 0.1600 \\
Contract 6 & 0.1456 & 0.1456 & 0.1453 & 0.1459 \\
Contract 7 & 0.1394 & 0.1394 & 0.1395 & 0.1410 \\
Contract 8 & 0.1329 & 0.1329 & 0.1334 & 0.1346 \\
Contract 9 & 0.1200 & 0.1199 & 0.1201 & 0.1217 \\
Contract 10 & 0.1091 & 0.1091 & 0.1092 & 0.1099 \\
Contract 11 & 0.0983 & 0.0984 & 0.0983 & 0.1007 \\
Contract 12 & 0.0910 & 0.0910 & 0.0913 & 0.0925 \\
Contract 13 & 0.0817 & 0.0817 & 0.0819 & 0.0832 \\
Contract 14 & 0.0700 & 0.0699 & 0.0699 & 0.0720 \\
Contract 15 & 0.0569 & 0.0569 & 0.0567 & 0.0579 \\
Contract 16 & 0.0489 & 0.0489 & 0.0488 & 0.0503 \\
Contract 17 & 0.0381 & 0.0381 & 0.0382 & 0.0391 \\
Contract 18 & 0.0292 & 0.0292 & 0.0294 & 0.0309 \\
Contract 19 & 0.0192 & 0.0192 & 0.0197 & 0.0206 \\
Contract 20 & 0.0077 & 0.0077 & 0.0074 & 0.0112 \\
Mean & 0.1026 & 0.1026 & 0.1026 & 0.1047 \\
\noalign{\smallskip}\hline\noalign{\smallskip}
\end{tabular}
\end{table}

\begin{table}
\caption{Root mean square errors (RMSEs) for futures contracts in the case of UKF: comparison of 13-contract data.}
\label{tab:RMSE_UKF1}  
\centering
\begin{tabular}{p{3cm}p{2cm}p{2cm}p{2cm}p{1cm}}
\hline\noalign{\smallskip}
Contracts & Case 1 & Case 2 & Case 3 & Case 4 \\
\noalign{\smallskip}\svhline\noalign{\smallskip}
Contract 1 & 0.1190 & 0.1219 & 0.1173 & 0.1173 \\
Contract 2 & 0.1131 & 0.1155 & 0.1132 & 0.1132 \\
Contract 3 & 0.1077 & 0.1097 & 0.1082 & 0.1082 \\
Contract 4 & 0.0938 & 0.0950 & 0.0930 & 0.0930 \\
Contract 5 & 0.0887 & 0.0898 & 0.0889 & 0.0889 \\
Contract 6 & 0.0773 & 0.0778 & 0.0771 & 0.0771 \\
Contract 7 & 0.0703 & 0.0711 & 0.0704 & 0.0705 \\
Contract 8 & 0.0619 & 0.0628 & 0.0621 & 0.0621 \\
Contract 9 & 0.0500 & 0.0506 & 0.0500 & 0.0500 \\
Contract 10 & 0.0400 & 0.0406 & 0.0399 & 0.0400 \\
Contract 11 & 0.0300 & 0.0311 & 0.0299 & 0.0300 \\
Contract 12 & 0.0200 & 0.0211 & 0.0201 & 0.0202 \\
Contract 13 & 0.0089 & 0.0111 & 0.0087 & 0.0087 \\
Mean & 0.0677 & 0.0691 & 0.0676 & 0.0676 \\
\noalign{\smallskip}\hline\noalign{\smallskip}
\end{tabular}
\end{table}

\begin{table}
\caption{Root mean square errors (RMSEs) for futures contracts in the case of UKF: comparison of 20-contract data.}
\label{tab:RMSE_UKF2}  
\centering
\begin{tabular}{p{3cm}p{2cm}p{2cm}p{2cm}p{1cm}}
\hline\noalign{\smallskip}
Contracts & Case 1 & Case 2 & Case 3 & Case 4 \\
\noalign{\smallskip}\svhline\noalign{\smallskip}
Contract 1 & 0.2375 & 0.2133 & 0.1868 & 0.1886 \\
Contract 2 & 0.2217 & 0.2085 & 0.1809 & 0.1820 \\
Contract 3 & 0.2252 & 0.1952 & 0.1820 & 0.1839 \\
Contract 4 & 0.1979 & 0.1804 & 0.1595 & 0.1607 \\
Contract 5 & 0.2024 & 0.1687 & 0.1614 & 0.1633 \\
Contract 6 & 0.1781 & 0.1631 & 0.1456 & 0.1466 \\
Contract 7 & 0.1734 & 0.1531 & 0.1414 & 0.1425 \\
Contract 8 & 0.1706 & 0.1426 & 0.1369 & 0.1384 \\
Contract 9 & 0.1539 & 0.1306 & 0.1225 & 0.1237 \\
Contract 10 & 0.1448 & 0.1179 & 0.1126 & 0.1139 \\
Contract 11 & 0.1319 & 0.1081 & 0.1014 & 0.1026 \\
Contract 12 & 0.1281 & 0.0975 & 0.0963 & 0.0977 \\
Contract 13 & 0.1138 & 0.0899 & 0.0853 & 0.0864 \\
Contract 14 & 0.1030 & 0.0788 & 0.0739 & 0.0750 \\
Contract 15 & 0.0911 & 0.0657 & 0.0609 & 0.0621 \\
Contract 16 & 0.0827 & 0.0584 & 0.0532 & 0.0544 \\
Contract 17 & 0.0774 & 0.0461 & 0.0454 & 0.0467 \\
Contract 18 & 0.0704 & 0.0373 & 0.0378 & 0.0392 \\
Contract 19 & 0.0624 & 0.0303 & 0.0294 & 0.0308 \\
Contract 20 & 0.0574 & 0.0223 & 0.0228 & 0.0244 \\
Mean & 0.1412 & 0.1154 & 0.1068 & 0.1081 \\
\noalign{\smallskip}\hline\noalign{\smallskip}
\end{tabular}
\end{table}

~

The RMSE for each contract is presented in Table~\ref{tab:RMSE_EKF1} and Table~\ref{tab:RMSE_UKF1} for the case of EKF and UKF algorithms respectively, based on 13-contract data. Additionally, Table~\ref{tab:RMSE_EKF2} and Table~\ref{tab:RMSE_UKF2} display the RMSE for the 20-contract data. It is easy to see that the RMSE values are similar across all four cases, regardless of the contract considered. This indicates that all futures contracts are estimated accurately. Furthermore, by assigning a larger measurement error to the short-term contract and a smaller measurement error to the long-term contract, the RMSE values decrease as the maturity time increases. Moreover, there is minimal difference in contract estimation when using either the EKF or UKF. The performance of both algorithms in terms of contract estimation is comparable, demonstrating their effectiveness in the context of this study.

\begin{table}
\caption{Estimated state parameters for the 13-contract data for Case 3 and Case 4. For Case 1 and Case 2 the state parameters are fixed to the true values.}
\label{tab:est_para1}  
\centering
\begin{tabular}{p{1.5cm}p{1.5cm}p{2cm}p{2cm}p{2cm}p{2cm}}
\hline\noalign{\smallskip}
~ & True & Case 3 - EKF & Case 4 - EKF & Case 3 - UKF & Case 4 - UKF \\
\noalign{\smallskip}\svhline\noalign{\smallskip}
$\kappa$         &  0.5 & 0.5166 & 0.9540 & 0.5169 & 0.5198 \\
$\gamma$         &  0.3 & 0.3151 & 0.4311 & 0.3153 & 0.3200 \\
$\mu_{\xi}$      &  1   & 1.8929 & -0.1651 & 1.8745 & 2.3371 \\
$\sigma_{\chi}$  &  1.5 & 1.4166 & 7.0949 & 1.4172 & 5.1781 \\
$\sigma_{\xi}$   &  1.3 & 1.4068 & 5.6780 & 1.4054 & 1.3884 \\
$\rho$           & -0.3 & -0.3025 & -0.2937 & -0.3019 & -0.2896 \\
$\lambda_{\chi}$ &  0.5 & 0.8498 & 7.8781 & 0.8537 & 2.0501 \\
$\lambda_{\xi}$  &  0.3 & 0.8423 & 4.5633 & 0.8195 & 1.2683 \\
$\chi_0$         &    0 & -1.6138 & 3.5164 & 0.1731 & -4.2461 \\
$\xi_0$          & 3.33 & 5.0634 & -9.3820 & 5.0790 & 3.6159 \\
\noalign{\smallskip}\hline\noalign{\smallskip}
\end{tabular}
\end{table}

\begin{table}
\caption{Estimated state parameters for the 20-contract data for Case 3 and Case 4. For Case 1 and Case 2 the state parameters are fixed to the true values. }
\label{tab:est_para2}  
\centering
\begin{tabular}{p{1.5cm}p{1.5cm}p{2cm}p{2cm}p{2cm}p{2cm}}
\hline\noalign{\smallskip}
~ & True & Case 3 - EKF & Case 4 - EKF & Case 3 - UKF & Case 4 - UKF \\
\noalign{\smallskip}\svhline\noalign{\smallskip}
$\kappa$         &  0.5 & 0.5115 & 0.4629 & 0.5122 & 0.5121 \\
$\gamma$         &  0.3 & 0.3083 & 0.4101 & 0.3092 & 0.3072 \\
$\mu_{\xi}$      &  1   & 1.7924 & -0.8217 & 1.8131 & 2.5670 \\
$\sigma_{\chi}$  &  1.5 & 1.4506 & 1.3957 & 1.4482 & 1.6945 \\
$\sigma_{\xi}$   &  1.3 & 1.3273 & 1.3459 & 1.3307 & 1.1172 \\
$\rho$           & -0.3 & -0.3046 & 0.3317 & -0.3037 & -0.2987 \\
$\lambda_{\chi}$ &  0.5 & 0.7988 & -0.7015 & 0.8082 & 1.0558 \\
$\lambda_{\xi}$  &  0.3 & 0.7986 & -1.8602 & 0.8081 & 0.8872 \\
$\chi_0$         &    0 & 0.7635 & 2.0328 & 0.9607 & -2.9934 \\
$\xi_0$          & 3.33 & 3.0004 & 2.0159 & 3.9452 & 6.5649 \\
\noalign{\smallskip}\hline\noalign{\smallskip}
\end{tabular}
\end{table}

\begin{table}
\caption{Estimated measurement errors for the 13-contract data for Case 3 and Case 4. For Case 1 and Case 2 the measurement errors are fixed to the true values. }
\label{tab:est_se1}  
\centering
\begin{tabular}{p{1.5cm}p{1.5cm}p{2cm}p{2cm}p{2cm}p{2cm}}
\hline\noalign{\smallskip}
~ & True & Case 3 - EKF & Case 4 - EKF & Case 3 - UKF & Case 4 - UKF \\
\noalign{\smallskip}\svhline\noalign{\smallskip}
$\sigma_1$    & 0.13 & 0.1255 & 0.1302 & 0.1255 & 0.1255 \\
$\sigma_2$    & 0.12 & 0.1207 & 0.1220 & 0.1207 & 0.1207 \\
$\sigma_3$    & 0.11 & 0.1133 & 0.1135 & 0.1133 & 0.1133 \\
$\sigma_4$    & 0.10 & 0.0966 & 0.0952 & 0.0967 & 0.0966 \\
$\sigma_5$    & 0.09 & 0.0925 & 0.0904 & 0.0925 & 0.0925 \\
$\sigma_6$    & 0.08 & 0.0790 & 0.0783 & 0.0790 & 0.0790 \\
$\sigma_7$    & 0.07 & 0.0721 & 0.0719 & 0.0721 & 0.0721 \\
$\sigma_8$    & 0.06 & 0.0630 & 0.0635 & 0.0630 & 0.0630 \\
$\sigma_9$    & 0.05 & 0.0507 & 0.0515 & 0.0507 & 0.0507 \\
$\sigma_{10}$ & 0.04 & 0.0404 & 0.0411 & 0.0404 & 0.0404 \\
$\sigma_{11}$ & 0.03 & 0.0298 & 0.0313 & 0.0298 & 0.0298 \\
$\sigma_{12}$ & 0.02 & 0.0207 & 0.0227 & 0.0207 & 0.0207 \\
$\sigma_{13}$ & 0.01 & 0.0100 & 1.03E-05 & 0.0100 & 0.0100 \\
\noalign{\smallskip}\hline\noalign{\smallskip}
\end{tabular}
\end{table}

\begin{table}
\caption{Estimated measurement errors for the 20-contract data for Case 3 and Case 4. For Case 1 and Case 2 the measurement errors are fixed to the true values. }
\label{tab:est_se2}  
\centering
\begin{tabular}{p{1.5cm}p{1.5cm}p{2cm}p{2cm}p{2cm}p{2cm}}
\hline\noalign{\smallskip}
~ & True & Case 3 - EKF & Case 4 - EKF & Case 3 - UKF & Case 4 - UKF \\
\noalign{\smallskip}\svhline\noalign{\smallskip}
$\sigma_1$    & 0.20 & 0.1953 & 0.2015 & 0.1953 & 0.1953 \\
$\sigma_2$    & 0.19 & 0.1902 & 0.1962 & 0.1903 & 0.1903 \\
$\sigma_3$    & 0.18 & 0.1869 & 0.1888 & 0.1869 & 0.1869 \\
$\sigma_4$    & 0.17 & 0.1646 & 0.1647 & 0.1646 & 0.1646 \\
$\sigma_5$    & 0.16 & 0.1636 & 0.1636 & 0.1636 & 0.1636 \\
$\sigma_6$    & 0.15 & 0.1487 & 0.1486 & 0.1486 & 0.1486 \\
$\sigma_7$    & 0.14 & 0.1433 & 0.1434 & 0.1433 & 0.1433 \\
$\sigma_8$    & 0.13 & 0.1366 & 0.1365 & 0.1366 & 0.1366 \\
$\sigma_9$    & 0.12 & 0.1225 & 0.1233 & 0.1225 & 0.1225 \\
$\sigma_{10}$ & 0.11 & 0.1109 & 0.1112 & 0.1109 & 0.1109 \\
$\sigma_{11}$ & 0.10 & 0.0995 & 0.1002 & 0.0995 & 0.0995 \\
$\sigma_{12}$ & 0.09 & 0.0924 & 0.0926 & 0.0925 & 0.0925 \\
$\sigma_{13}$ & 0.08 & 0.0826 & 0.0834 & 0.0826 & 0.0826 \\
$\sigma_{14}$ & 0.07 & 0.0706 & 0.0708 & 0.0706 & 0.0706 \\
$\sigma_{15}$ & 0.06 & 0.0570 & 0.0574 & 0.0570 & 0.0570 \\
$\sigma_{16}$ & 0.05 & 0.0493 & 0.0498 & 0.0493 & 0.0493 \\
$\sigma_{17}$ & 0.04 & 0.0387 & 0.0389 & 0.0388 & 0.0387 \\
$\sigma_{18}$ & 0.03 & 0.0294 & 0.0294 & 0.0294 & 0.0294 \\
$\sigma_{19}$ & 0.02 & 0.0208 & 0.0204 & 0.0208 & 0.0208 \\
$\sigma_{20}$ & 0.01 & 0.0093 & 0.0113 & 0.0093 & 0.0093 \\
\noalign{\smallskip}\hline\noalign{\smallskip}
\end{tabular}
\end{table}

%\begin{table}
%\caption{Estimated coordinate representations (EKF). }
%\label{tab:est_coe_EKF}  
%\centering
%\begin{tabular}{p{1.5cm}p{2cm}p{2cm}p{2cm}p{2cm}p{1cm}}
%\hline\noalign{\smallskip}
%~ & True & Case 1 & Case 2 & Case 3 & Case 4 \\
%\noalign{\smallskip}\svhline\noalign{\smallskip}
%$\alpha_1$ & 5 & NA & 4.5149 & NA & -3.8661 \\
%$\alpha_2$ & 2 & NA & 1.8340 & NA & -2.2637 \\
%$\alpha_3$ & 2 & NA & 2.3138 & NA & -9.9986 \\
%$\alpha_4$ & 2 & NA & 1.9874 & NA & 0.0163 \\
%$\alpha_5$ & 3 & NA & 3.0488 & NA & 0.0470 \\
%$\alpha_6$ & 1 & NA & 0.9661 & NA & -0.7228 \\
%\end{tabular}
%\end{table}

%\begin{table}
%\caption{Estimated coordinate representations (UKF). }
%\label{tab:est_coe_UKF}  
%\centering
%\begin{tabular}{p{1.5cm}p{2cm}p{2cm}p{2cm}p{2cm}p{1cm}}
%\hline\noalign{\smallskip}
%~ & True & Case 1 & Case 2 & Case 3 & Case 4 \\
%\noalign{\smallskip}\svhline\noalign{\smallskip}
%$\alpha_1$ & 5 & NA & 10 & NA & 3.8891 \\
%$\alpha_2$ & 2 & NA & 8.5185 & NA & -0.2656 \\
%$\alpha_3$ & 2 & NA & -2.2797 & NA & 0.3508 \\
%$\alpha_4$ & 2 & NA & 1.8713 & NA & 0.1451 \\
%$\alpha_5$ & 3 & NA & -3.6190 & NA & 0.8563 \\
%$\alpha_6$ & 1 & NA & 0.5497 & NA & 1.0349 \\
%\end{tabular}
%\end{table}
\begin{table}
\caption{Estimated coordinate representations for the 13-contract data for Case 2 and Case 4. For Case 1 and Case 3 the coordinate representations are fixed to the true values.}
\label{tab:est_coe1}  
\centering
\begin{tabular}{p{1.5cm}p{1.5cm}p{2cm}p{2cm}p{2cm}p{2cm}}
\hline\noalign{\smallskip}
~ & True & Case 2 - EKF & Case 4 - EKF & Case 2 - UKF & Case 4 - UKF \\
\noalign{\smallskip}\svhline\noalign{\smallskip}
$\alpha_1$ & 5 & 4.5149 & -3.8661 & 10 & 3.8891 \\
$\alpha_2$ & 2 & 1.8340 & -2.2637 & 8.5185 & -0.2656 \\
$\alpha_3$ & 2 & 2.3138 & -9.9986 & -2.2797 & 0.3508 \\
$\alpha_4$ & 2 & 1.9874 & 0.0163 & 1.8713 & 0.1451 \\
$\alpha_5$ & 3 & 3.0488 & 0.0470 & -3.6190 & 0.8563 \\
$\alpha_6$ & 1 & 0.9661 & -0.7228 & 0.5497 & 1.0349 \\
\noalign{\smallskip}\hline\noalign{\smallskip}
\end{tabular}
\end{table}

\begin{table}
\caption{Estimated coordinate representations for the 20-contract data for Case 2 and Case 4. For Case 1 and Case 3 the coordinate representations are fixed to the true values. }
\label{tab:est_coe2}  
\centering
\begin{tabular}{p{1.5cm}p{1.5cm}p{2cm}p{2cm}p{2cm}p{2cm}}
\hline\noalign{\smallskip}
~ & True & Case 2 - EKF & Case 4 - EKF & Case 2 - UKF & Case 4 - UKF \\
\noalign{\smallskip}\svhline\noalign{\smallskip}
$\alpha_1$ & 5 & 5.3698 & 8.0156 & 6.7448 & 8.9205 \\
$\alpha_2$ & 2 & 2.5473 & -4.2205 & 8.5185 & -5.7008 \\
$\alpha_3$ & 2 & 1.9943 & 6.7163 & -2.2797 & -4.9896 \\
$\alpha_4$ & 2 & 2.0308 & 2.1835 & 1.8713 & 1.4698 \\
$\alpha_5$ & 3 & 2.8492 & 2.1582 & -3.6190 & 2.9955 \\
$\alpha_6$ & 1 & 0.9686 & -2.3905 & 0.5497 & 1.4239 \\
\noalign{\smallskip}\hline\noalign{\smallskip}
\end{tabular}
\end{table}

However, when we look at the parameter estimations, the conclusion changes. Table~\ref{tab:est_para1} - Table~\ref{tab:est_coe2} give the estimated state parameters, measurement errors and coordinate representations, respectively. For the state parameters, the estimations in Case 4 (coordinate representations are estimated) change a lot compared to the estimations in Case 3 (coordinate representations are fixed), which suggests the existence of parameter identification problems. The trends of state variables are somehow captured by the coordinate representations. Moreover, the estimations in Case 3 are closer to the true values. 

For the measurement errors (Table~\ref{tab:est_se1} and Table~\ref{tab:est_se2}), there are not many differences between Case 3 and Case 4. Measurement errors are estimated accurately. There are also not many differences between EKF and UKF. 

The estimated coordinate representations are given in Table~\ref{tab:est_coe1} and Table~\ref{tab:est_coe2}. Comparing Case 2 (state parameters and measurement errors are fixed) to Case 4 (state parameters and measurement errors are estimated), like state parameters, the estimations of coordinate representations change a lot. Moreover, in Case 2, the estimations of coordinate representations filtered by EKF are close to the true values, but not the estimations filtered by UKF are not. In that case, coordinate representations can be estimated through EKF. In Case 4, both EKF and UKF cannot estimate coordinate representations.  

\section{Conclusion}
\label{sec:conclusion}

In the modelling of commodity futures, it is common to assume that the logarithm of the underlying spot price is expressed  as a sum of various factors. However, this class of models is subject to two limitations. Firstly, for deriving a closed-form expression for the futures price, the spot price used to be a linear function of Gaussian distributed factors. Secondly, these models imply that the spot price always be positive. To overcome these two limitations, we introduced the polynomial diffusion model in this paper, which serves as a generalisation of the Schwartz-Smith two-factor model. This model allows for a more flexible and non-linear representation of the spot price. Specifically, we applied a polynomial diffusion model of degree 2 to the two-factor model. The estimation of parameters and hidden state variables was performed using the EKF and the UKF. 

We conducted a study to assess the performance of the models in four different cases using the simulated data. Overall, we found that the futures contracts can be accurately estimated. However, parameter estimation remains challenging, even when we impose constrains on the model parameters and estimate coordinate representations of the polynomial of spot price. In other words, while the state variables $\chi_t$ and $\xi_t$ can not be estimated correctly, however the futures prices were recovered reasonably well. This phenomenon suggests the existence of identification problems. Firstly, the trends of state variables are captured by the spot price polynomials. As a consequence, the estimates of both state parameters and coordinate representations of the polynomial are far away from the true values. Secondly, even when we fix one set of parameters (either the model parameters or coordinate representations), the estimation of the other set of parameters still fails to converge to the true values. The identification issue poses the challenge of finding the correct order for the polynomial diffusion process. Further research on the parameter identification problem is needed to address this issue.

% BibTeX users please use
\bibliographystyle{spmpsci}
\bibliography{References}

\end{document}